\documentclass{wine-class/llncs}
\usepackage{color}
\usepackage{graphicx}
\usepackage{tikz}
\usetikzlibrary{patterns, arrows, decorations.markings}
\usepackage{caption, subcaption}
\usepackage{color}
\usepackage{listings}
\usepackage[most]{tcolorbox}
\tcbuselibrary{theorems}
\usepackage{algorithm, algorithmic}
\usepackage{hyperref}
\usepackage{media9}
\usepackage{comment}

\usepackage{mathrsfs, mathtools}

\author{Victor Boone$^*$ \and Georgios Piliouras$^\dagger$}
\institute{\hspace{-32pt}$^*$ENS de Lyon  $\hspace{48pt}^\dagger$SUTD}
\title{From Darwin to Poincar\'e and von Neumann: Recurrence and Cycles in Evolutionary and Algorithmic Game Theory}

\newcommand\supp{{\textnormal{supp}}}
\newcommand\interior{{\textnormal{int}}}
\newcommand\divergence{{\textnormal{div}}}
\newcommand\dist{{\textnormal{dist}}}
\newcommand\kl{{\textnormal{KL}}}
\newcommand\R\bbbr
\newcommand\N\bbbn

\newcommand\f{{\boldsymbol f}}
\newcommand\vol{{\textnormal{vol}}}
\newcommand\varnothing{{\textnormal{\footnotesize O}\hspace{-6.75pt}\boldsymbol/}}

\begin{document}
	\maketitle

  \begin{abstract}
Replicator dynamics, the continuous-time analogue of Multiplicative Weights Updates, is the main dynamic in evolutionary game theory. 
In simple evolutionary zero-sum games, such as Rock-Paper-Scissors, replicator dynamic is  periodic \cite{zeeman1980population}, however, its behavior in higher dimensions is not well understood.
We provide a complete characterization of its behavior  in zero-sum evolutionary games. We prove that, if and only if, the system has an interior Nash equilibrium, the dynamics exhibit Poincar\'{e} recurrence, i.e., almost all orbits come arbitrary close to their initial conditions infinitely often. If no interior equilibria exist, then all interior initial conditions converge to the boundary. Specifically, the strategies that are not in the support of any equilibrium vanish in the limit of all orbits. All recurrence results furthermore extend to a class of games that generalize both graphical polymatrix games as well as evolutionary games, establishing a unifying link between evolutionary and algorithmic game theory. We show that  two degrees of freedom, as in Rock-Paper-Scissors, is sufficient  to prove periodicity.
  \end{abstract}

\vspace{-3em}

\begin{figure}[H]
  \centering
  \includegraphics[width=.8\linewidth]{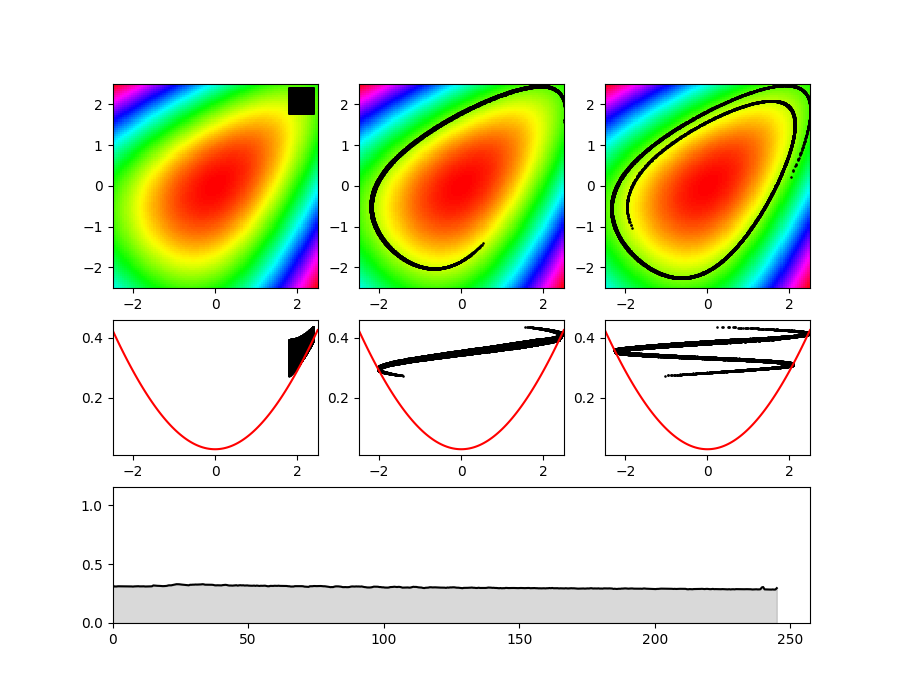}
  \caption{One agent Rock-Paper-Scissor,
  see \ref{front-diagram}. Animation \href{https://www.dropbox.com/s/c37cgiztlcryje6/foo.avi?dl=0}{here}}
\end{figure}

\section{Introduction}


Replicator dynamics is a basic model of evolution that is amongst the most well studied game theoretical models of adaptive behavior \cite{Taylor1978145,Schuster1983533}. It is the standard dynamic in evolutionary game theory \cite{Hofbauer96,Wei95} and enjoys formal connections to other classic evolutionary models such as the Price equation, the Lotka-Volterra equation of ecology and the quasispecies equation of molecular evolution \cite{page2002unifying,bomze1995lotka}. Replicator also has strong, inherent connections to computer science and optimization theory. It is the continuous time-analogue of Multiplicative Weights Update \cite{Kleinberg09multiplicativeupdates}, arguably the most widely studied online learning and optimization algorithm and a meta-algorithmic  technique in itself with numerous applications \cite{Arora05themultiplicative}.
Furthermore, it has diverse microcanonical foundations \cite{Sandholm10}, i.e., it can emerge from numerous, simple (memoryless, best-response like) population dynamics, which enhance its plausibility as a model of emergent behavior.
Finally, it has an interpretation as an inference dynamic
 \cite{harper2011escort,karev2010replicator}. Specifically, 
 for systems governed by the replicator equations the maximum entropy principle (MaxEnt) 
 can be derived  rather than postulated as, e.g., in thermodynamics or statistical mechanics \cite{jaynes1957information}. Given this impressive web of connections, it would not be unreasonable to think of replicator as a near-universal model of adaptive behavior, a proto-intelligence mechanism, emerging from simple physical processes and giving rise to self-organizing, ever-more complex and efficient systems. 
  As such understanding its behavior in different contexts  can simultaneously shed light to many of its related adaptive processes.

Evolution as it turns out is a very efficient force of systemic optimization.
Replicator dynamics is a regret minimizing dynamic in arbitrary games. Its regret converges to zero at a rate of $O(1/T)$ \cite{mertikopoulos2017cycles,Kwon}. Specifically, its total regret remains bounded for all time. In cases of games where agents' interests are strongly aligned, such as potential, i.e., congestion games, replicator dynamics is known to perform admirably well. Not only does it converge to Nash equilibria \cite{sandholm2008projection} but typically to pure Nash \cite{Kleinberg09multiplicativeupdates}. Furthermore, it has been shown that pure Nash equilibria of higher social welfare have larger regions of attraction and hence an average case analysis of replicator dynamics where the initial condition is drawn uniformly at random can lead to an expected social welfare that can be much higher than those predicted by Price of Anarchy analysis \cite{panageas2016average}. Finally, even in games where the dynamics are non-equilibrating  replicator dynamics may converge to limit cycles with optimal social welfare that dominate the performance of even the best Nash equilibrium by an arbitrary amount \cite{paperics11,gaunersdorfer1995fictitious}. That is, replicator dynamic can significantly outperform even Price of Stability type of guarantees. 

When we move to zero-sum games (and variants thereof) Price of Anarchy and more generally social welfare optimization type of results are no longer applicable. One would hope that in such games the Nash equilibria would be accurate predictors of the system behavior. If so equilibration would have not only a strong economic and algorithmic justification due to the celebrated maxmin theorem by von Neumann \cite{Neumann1944} and its connection to linear programming but also an evolutionary one. Unfortunately, this is not the case. \cite{Sato02042002} established experimentally that even small zero-sum games may have complex, non-equilibrating, chaotic type of behavior. More recently, \cite{piliouras2014optimization} established that despite their chaotic behavior, these dynamics have also exploitable structure. Specifically, replicator dynamics in two-player zero-sum games with interior Nash equilibria are Poincar\'{e} recurrent. This means that almost all initial conditions return infinitely often arbitrarily close to their initial conditions. This result holds even for networks of zero-sum games, however, this class of games fails to capture the standard class of evolutionary zero-sum games. The immediate distinction between evolutionary games and standard multi-agent games is that evolutionary games only admit a single distribution over a simplex of strategies.
 These are games where a large population of animals  
  compete against each other and where the frequencies of the different genotypes/strategies evolve according to the replicator dynamics. From the perspective of standard two-agent zero-sum games, the question reduces to analyzing antisymmetric zero-sum games (i.e. Rock-Paper-Scissors) under symmetric initial conditions. Due to the (anti)-symmetric nature of the game, the symmetry of initial condition is preserved by the dynamic. Thus, the dynamic evolves on a lower-dimensional manifold, which is a zero-measure set, hence the Poincar\'{e} recurrence result of  \cite{piliouras2014optimization} does not suffice to understand the behavior for such non-generic initial conditions. Our goal here is to completely understand the behavior of replicator dynamics in such settings and furthermore develop an expansive unifying framework for understanding dynamics both  in evolutionary games as well as  two-agent and multi-agent settings as well.

{\bf Our results.} We provide a complete characterization of the behavior of replicator dynamic in zero-sum evolutionary games. We prove that if and only if, the system has an interior Nash equilibrium, the dynamics exhibit Poincar\'{e} recurrence. If no interior equilibria exist, then all interior initial conditions converge to the boundary (Theorem \ref{collapses}). Specifically, the strategies that are not in the support of any equilibrium vanish in the limit of all orbits. All recurrence results furthermore extend to a class of games that generalize both graphical polymatrix games as well as evolutionary games (Theorem \ref{general-thm}). Specifically, we allow for polymatrix edges with self-edges, where all polymatrix games are constant-sum, and all self-edges are antisymmetric games.  To prove these results, we provide the most general to date set of game theoretic conditions under which replicator dynamics can be shown to be volume preserving (under a diffeomorphism, i.e. a differentiable transformation with invertible inverse) (Theorem \ref{coro-vol}). The other stepping stone in the direction of proving recurrence/convergence to the boundary is showing that the KL-divergence between the Nash equilibrium and the state of the system is invariant/strictly decreasing if the zero-sum games has/(does not have) an interior Nash. This argument mirrors arguments for the case of multiple agent replicator dynamics \cite{piliouras2014optimization,mertikopoulos2017cycles}
 Finally, we show that in this class of games, two degrees of freedom, as in Rock-Paper-Scissors, is sufficient  to prove periodicity (Theorem \ref{periodic}). Furthermore, as we argue this does not follow from an immediate combination of Poincar\'{e} recurrence and Poincar\'{e}-Bendixson theorems but requires more specialized arguments. The full version of this paper can be found online \cite{Boone19}.



\section{Related work}

{\bf Non-equilibration, recurrence and volume preservation.} 
 In evolutionary game theory, numerous non-convergence results are known but they are usually restricted to small games \cite{Sandholm10}. 
 \cite{Akin84} was the first paper to study both discrete and continuous-time evolutionary dynamics in zero-sum games and establish invariant for the dynamics, however, no formal recurrence or periodicity was shown. Constants of the motion exist for different classes of games (e.g. coordination/partnership games, null stable games) and dynamics \cite{HS98,Sandholm10,panageas2016average} even for games with convergent dynamics. An orthogonal property of game dynamics is the preservation  of volume of initial conditions (up to state space/speed transformation, see \cite{HS98,Sandholm10,mertikopoulos2017cycles}). 
  \cite{piliouras2014optimization} and \cite{PiliourasAAMAS2014} showed that replicator dynamics in (network) zero-sum games (and affine variants thereof) exhibit a specific type of repetitive behavior, known as Poincar\'{e} recurrence by combining these two type of arguments.  
  Recently, \cite{mertikopoulos2017cycles} proved that Poincar\'{e} recurrence also shows up in a more general class of continuous-time dynamics known as Follow-the-Regularized-Leader (FTRL).  \cite{2017arXiv171011249M} established that the recurrence results for replicator dynamics extend to some biologically-inspired dynamically evolving zero-sum games.
  Perfectly periodic (i.e., cyclic) behavior  for replicator may arise in team competition \cite{DBLP:journals/corr/abs-1711-06879} as well as in network competition \cite{nagarajan2018three}. 
 Our techniques build and extend upon these results by producing necessary, as well as sufficient conditions, for volume preservation, recurrence as well as periodicity.

{\bf Game dynamics as physics.} 
Recently, \cite{BaileyAAMAS19} established a connection between game theory, online optimization in continuous-time (FTRL dynamics) and a ubiquitous class of systems in physics known as Hamiltonian dynamics, which exhibit conservation laws (``conservation of energy").  
 In the case of discrete-time dynamics such as MWU or gradient descent the system trajectories are first order approximations of the continuous-time dynamics. Energy conservation and recurrence no longer hold. Instead  energy increases and the dynamics divergence to the boundary \cite{BaileyEC18}. The dynamics exhibit volume expansion and Lyapunov chaos \cite{CP2019}. Despite this divergent, chaotic behavior, gradient descent with fixed step size, has vanishing regret in small zero-sum games \cite{2019arXiv190504532B}.  More elaborate discretization techniques, based on leap-frogging (Verlet) symplectic integration technique for Hamiltonian dynamics, result in discrete-time algorithms of bounded regret in general games and Poincar\'{e} recurrence in zero-sum games respectively \cite{2019arXiv190704392B}.
  So far, it is not clear to what extent the connections with Hamiltonian dynamics can be generalized; however, \cite{ostrovski2011piecewise} have considered a class of piecewise affine Hamiltonian vector fields whose orbits are piecewise straight lines and developed the connections with best-reply dynamics.  


{\bf Game dynamics as dynamical systems.}
Finally, \cite{Entropy18,papadimitriou2019game} initiated a program for linking game theory to topology, specifically to Conley's fundamental theorem of dynamical systems \cite{conley1978isolated}. 
 This approach shifts attention from Nash equilibria to a more general notion of recurrence, called chain recurrence, that generalizes both periodicity and Poincar\'{e} recurrence. 
\cite{omidshafiei2019alpha} embeds this approach within an algorithmically tractable framework and uses it to develop new training algorithms for multi-agent AI settings.

\section{Preliminaries and definitions}

\subsection{Zero-sum games and Zero-sum polymatrix games}

A \emph{graphical polymatrix game} is defined using a directed graph $G = 
(V, E)$ where $V$ corresponds to the set of agents (or players) 
and where every edge corresponds to a \emph{bimatrix game} between its two 
endpoints/agents. Each agent $i \in V$ has a set of \emph{actions} $\mathcal A_i 
=\{1 \ldots n_i\}$ that he is allowed to select randomly under a distribution
$x_i$ called a \emph{mixed stragegy}. The set of mixed strategies of player 
$i$ is written $\mathcal X_i = \Delta \R^{n_i} = \{x_\alpha \in \R^{n_i}_{\ge 0}
\text{ : } \sum_\alpha x_\alpha = 1\}$; the state of the game is then
defined by the concatenation of strategies of all players. We call 
\emph{strategy space} the set of all possible strategies profiles, and write
it $\mathcal X \equiv \prod_{i \in V} \mathcal X_i$.

The bimatrix game on edge $(i, j)$ is described using a pair of matrices 
$A^{i, j} \in \R^{|\mathcal A_i| \times |\mathcal A_j|}$ and $A^{j, i} \in 
\R^{|\mathcal A_j| \times |\mathcal A_i|}$. The coefficient $(\alpha, \beta)
\in \mathcal A_i \times \mathcal A_j$ of the matrix $A^{i, j}$ represents the 
reward player $i$ gets when he plays $\alpha$ against player $j$ playing $\beta$.
As players can choose mixed strategies, their payoffs are random variables, yet
we call \emph{payoffs} again their expected payoffs. For instance, the payoff
of player $i$ against player $j$ is $x_i \cdot A^{i,j} x_j$.
We call \emph{payoff of agent $i \in V$} under strategy profile
$x$ the sum of the payoffs agent $i$ receives from every bimatrix
game he participates in, and write it $u_i(x)$ or  $u_i(x_i; x_{-i})$. More precisely, 
\begin{equation}
u_i(x) = \sum_{j \text{ : } (i, j) \in E} x_i \cdot A^{i, j} x_j
\end{equation}
Sometimes, one can be interested in the payoff of agent $i$ when deviating 
to action $\alpha \in \mathcal A_i$ under profile $x$. This quantity is usually 
denoted $u_{i,\alpha}(x)$ and corresponds to $\sum_{j=1}^N (A^{i, j} x_j)_{
\alpha}$. Finally, we will compactify the definition of a $N$-player graphical 
polymatrix game by a tuple $\Gamma = (G, A)$ with $G$ the underlying graph and 
$A$ the block matrix built from $A^{i, j}$'s. 

We say that a $N$-player graphical polymatrix game is \emph{zero-sum}
if the matrix $A$ is antisymmetric. In the case $N=2$ players $i, j = 1, 2$, 
it specifically means that $A^{1,2} = - (A^{2, 1})^T$; in the case $N=1$ player,
that $A^{1, 1} = A$ is antisymmetric. In our case, we allow the graph 
$G$ to contain \emph{self-loops}, and we call \emph{diagonal games} the subgames
induced by self-loops. Self-loop (1-agent) games make sense both in the content
of evolutionary game theory as well as in classic (multi-agent) game theory.
From the perspective of evolutionary game theory, 1-agent games are the norm 
where we study the frequencies of different competing genotypes within a single 
population. For example, Rock-Paper-Scissors could be different traits that
exhibit a cyclic pattern of dominance. In the context of classic game theory 
a single agent self-loop added e.g. on top of a standard normal form game can 
capture effects like friction in dynamics, e.g. the matrix with zero diagonal 
and -1 in all other entries captures the effects of having cost for changing 
strategy. Specifically, if an agent changes her strategy from yesterday, then 
in the self-loop game, she experiences an additional cost of 1. 
More generally, it allows to differentiate the performance of a strategy for an
agent depending on his strategy in the previous time period in game dynamics such
as replicator dynamics.

\begin{figure}[H]
  \centering
  \begin{tikzpicture}
    \node[circle, draw] (A) at (0, 1) {$P_0$};
    \draw[->, >=latex] (A) to[min distance=8mm, in=90-30, out=90+30] (A);
    
    \node[circle, draw] (A) at (-.75+3,0.5) {$P_1$};
    \node[circle, draw] (B) at (.75+3, 1.5) {$P_2$};
    \draw[->, >=latex] (A) to[bend right] (B);
    \draw[->, >=latex] (B) to[bend right] (A);

    \node[circle, draw] (A) at (0+6, 0) {$P_1$}; 
    \node[circle, draw] (B) at (3+6, .5) {$P_2$}; 
    \node[circle, draw] (C) at (1+6, 2) {$P_3$}; 
    
    \draw[->, >=latex] (A) to[bend right] (B);
    \draw[->, >=latex] (B) to (A);
    \draw[->, >=latex] (C) to (B);
    \draw[->, >=latex] (A) to[min distance=8mm, in=-140-30, out=-140+30] (A);
    \draw[->, >=latex] (B) to[min distance=8mm, in=-10+30, out=-10-30] (B);
  \end{tikzpicture} 
  \caption{From left to right, graphical representations of the evolutionary
  game setting, algorithmic game theory, and the merger of the two.}
\end{figure}
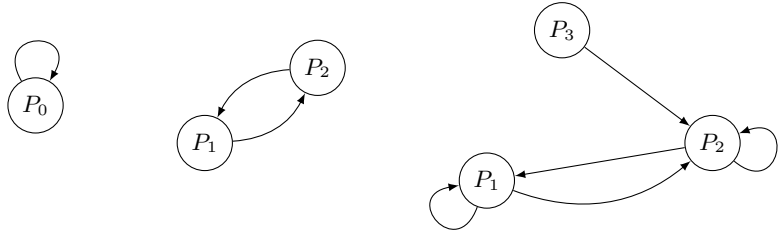

A very common notion in game theory is the one of \emph{Nash equilibrium} (NE),
defined in our case as a mixed strategy profile $x^* \in \mathcal X$ 
such that 
\begin{equation}
  u_i(x^*) \ge u_{i,\alpha_i}(x^*)
\end{equation}
for every strategy $\alpha_i \in \mathcal A_i$ of any player $i \in \mathcal N$.
We write $\supp (x_i^*) \equiv \{\alpha_i \in \mathcal A_i \text{ : }
x_{i, \alpha_i} > 0\}$ the support of $x_i^* \in \mathcal X_i$. A Nash 
equilibrium is said
\emph{interior} or \emph{fully mixed} if $\supp (x_i^*)$  for each agent $i$ is $\mathcal A_i$.

\subsection{Replicator dynamics}

The \emph{replicator equation} is one of the most well studied evolutionary processes. Its most usual formulation is:
\begin{equation}
  \label{rd}
  \dot x_{i, \alpha} = \frac{dx_{i,\alpha}}{dt}
  = x_{i, \alpha} \left(u_{i, \alpha}(x) - u_i(x)\right) 
\end{equation}
for every player $i$ and action $\alpha \in \mathcal A_i$. We will often 
translate (\ref{rd}) into \emph{cumulative costs space} via the diffeomorphism from
the interior of $\mathcal X$ to $\mathcal C \equiv \prod \R^{n_i-1}$,
used also in \cite{piliouras2014optimization}, that,
for each player $i$, maps $x_i = (x_{i,1} \ldots x_{i, n_i})$ to
$(\ln \frac{x_{i,2}}{x_{i, 1}} \ldots \ln \frac{x_{i, n_i}}{x_{i,1}})$.
We will write this diffeomorphism $\f$ and its inverse $\f^{-1}$. $\mathcal C$
is called this way since one can show that it corresponds to the space
of coordinates $\int_0^t u_{i,\alpha}(x(\tau)) d\tau$ up to a re-centralization
term (specifically,
$\int_0^t u_{i,\alpha}(x(\tau)) d\tau - \int_0^t u_{i,1}(x(\tau))d\tau$
for $\alpha > 1$). 

\subsection{Topology of dynamical systems}

\paragraph{Flows.} Since the strategy space is compact and the replicator
dynamics Lipschitz-continuous, there exists a continuous function 
$\phi : \mathcal X \times \R \rightarrow \mathcal X$ called \emph{flow of 
replicator dynamics (\ref{rd})} such that for any point $x \in \mathcal X$, 
$\phi(x, -)$ defines a function of time corresponding to the trajectory of $x$. 
Conversely, fixing a time $t$ provides a map $\phi^t \equiv \phi(-, t) :
\mathcal X \rightarrow \mathcal X$, and the family 
$\{\phi^t \text{ : } t \in \R\}$ is interestingly a subgroup of
$(\mathcal C(\mathcal X, \mathcal X), \circ)$. Moreover, if $\phi^t : 
A \rightarrow A$ and $\psi^t : B \rightarrow B$ are flows such that there
exists a diffeomorphism $g$ satisfying
$g(\phi^t(x)) = \psi^t(g(x))$ for all $x \in A$, then $\phi^t$ and $\psi^t$
are said to be \emph{diffeomorphic} to each other. 

\paragraph{Limit sets.} When $x \in \mathcal X$ is not a rest point of 
(\ref{rd}), we wish to grasp how the orbit of $x$ will asymptotically behave. 
In general,
its trajectory will not converge to a single point, but  to a  closed
set called the \emph{$\omega$-limit (set) of $x$}, written $\omega(x)$. 
This set is formally defined as the set of points $y \in \mathcal X$ such that
there exists a sequence $(t_n)$ diverging to $+\infty$ such that 
$\phi(x, t_n) \rightarrow y$. One alternative definition is 
$\omega(x)= \bigcap_{t \ge 0} \overline{\bigcup_{\tau \ge t}
\phi(x, \tau)}$. The compactness of $\omega(x)$ is an immediate consequence 
of the compactness of $\mathcal X$, and $\lim_{t \rightarrow +\infty} 
\dist(\phi(x, t), \omega(x)) = 0$.

\paragraph{Liouville's formula.} Liouville's formula can be applied to any
system of ordinary differential equations with a continuously differentiable
vector field $\xi$ on an open domain $\mathcal X \subseteq \R^d$. The divergence
of $\xi$ at $x \in \mathcal X$ is the trace of the Jacobian at $x$, that is,
$\divergence\; \xi(x) = \sum_{i=1}^d \frac{\partial \xi_i}{\partial x_i}(x)$.
Because the divergence is continuous, it is integrable on measurable subsets 
of $\mathcal X$. Given any such set $A$, define the image of $A$ under the 
flow $\phi$ at time $t$ as $A(t) = \{\phi(a, t) \text{ : } a \in A\}$. 
$A(t)$ is measurable and of volume $\vol[A(t)] = \int_{A(t)} d\mu$.
Liouville's formula states that the time derivative of the volume
$\vol[A(t)]$ exists and links it to the divergence of $\xi$,
\begin{equation}
  \label{liouville}
  \frac d{dt}\bigg[\vol \; A(t)\bigg] = \int_{A(t)} \divergence(\xi) \;d\mu
\end{equation}
One immediate consequence is that if $\divergence \;\xi(x)$ is null at any
$x \in \mathcal X$, then the volume is conserved. As $\divergence \; \xi$ is
clearly a continuous function, the reverse statement is also true. If the 
volume is preserved on any open set, $\divergence \;\xi(x)$ has to be null
at any point $x \in \mathcal X$.

\paragraph{Poincar\'e recurrence.}
This paper is focused on a recurrence behavior introduced by Poincar\'e and
more precisely by his studies on the three body problem. He proved 
\cite{Poincare1890} that 
as soon as a dynamical system preserves volume and that every orbit
remains bounded, almost all trajectories return
arbitrarily close to their initial position, and do so infinitely often.

\begin{theorem}[Poincar\'e recurrence] \cite{barreira}
If a flow preserves volume and has only bounded orbits then for each open set,
almost all orbits intersecting the set intersect it infinitely often.  
\end{theorem}

\subsection{Volume conservation and periodicity in Rock-Paper-Scissors}

\label{front-diagram}

The front page figure shows the evolution of a set of initial conditions (black square) under replicator
dynamics (\ref{rd}) in the (projected)  cumulative payoff space; the 
game is the classic one agent Rock-Paper-Scissors with payoff matrix
\begin{equation*}
\begin{pmatrix}
  0 & 1 & -1 \\ -1 & 0 & 1 \\ 1 & -1 & 0
\end{pmatrix}
\end{equation*}
In the first row of the figure
from left to right, we plot the evolution of
a set at times $t = 0$,
$t = 112$ and $t = 225$. The colormap represents the Kullback-Leibler divergence
to the unique Nash equilibrium $x = (\frac 13, \frac 13, \frac 13)$ (null 
vector in cumulative payoffs space). Observe that any point stays at the same
color at which it started, i.e., its Kullback-Leibler divergence from the Nash
equilibrium does not change. On the second row are the corresponding 
plots of the Kullback-Leibler divergence of points ($y$-axis) according to their 
first coordinate ($y_1$ in $\mathcal C$). The red curve is the minimum possible
value for each $y_1$ values, that is, an analogue of \emph{potential energy}.
Intuitively, any initial condition will slide along an horizontal level set
(of constant $\kl$-divergence) and cannot escape outside the red curve.
The third row shows an estimation of the volume of the cloud of points over 
time. This volume is estimated using a pruned Delaunay triangulation, more 
precisely, triangles with a diameter larger than some threshold value are 
deleted, and the volume is computed as the sum of the volume of each remaining
triangle.

Even though the shape of the initial condition is not preserved, the 
overall volume is  constant over time. This spiralling snake shape results
from periodic orbits of different periods.

\section{Volume conservation: Necessary and sufficient conditions}

\subsection{Zero-sum games are volume conservative}

Replicator dynamics in multi-games (with no loops) are volume
conservative (even beyond replicator dynamics \cite{mertikopoulos2017cycles}). 
The reason becomes clear when we examine the differential equation 
satisfied by cumulative payoffs $y_{i,\alpha}$'s.
\begin{equation}
  \frac{dy_{i,\alpha}}{dt}(t) = u_{i, \alpha+1}\left(\f^{-1}(y)\right) - 
  u_{i, 1}
  \left(\f^{-1}(y)\right)
\end{equation}
Recall that $\f^{-1}$ acts like a set-wise product function, working locally
at each player. Hence, as long as $u_{i, \beta}$ does not depend on $x_i$, 
$u_{i, \beta}$ is independent of $y_{i, \alpha}$ for any  pair of actions $\alpha, \beta$. The partial derivative
$\frac{\partial \dot y_{i, \alpha}}{\partial y_{i, \alpha}}(y)$ is null. One 
can understand this as follows: \emph{if the performance of an action only depends
on the behavior of the rest of the agents, then the volume is preserved.} 
In single agent games, antisymmetry implies volume preservation \cite{Sandholm10}.
We show that these results can be combined and that there is no need to have null diagonal games to get volume 
preservation in multi-agent games. The zero-sum property is enough to guarantee it.

\begin{theorem}
  \label{volcons-general}
  Let $\phi$ be the flow of replicator dynamics (\ref{rd}) with $N$ agents. 
  Let $\psi(y, -) = \f(\phi(\f^{-1}(y), -))$ be the diffeomorphic flow onto 
  cumulative payoffs space. If all diagonal games are zero-sum, then
  $\psi$ is volume conservative.
\end{theorem}

We know that if there are no games on the diagonal, the volume is preserved.
The intuition is that if each diagonal game preserves volume individually,
there will be volume preservation; this is the main point. We rely on
Liouville's formula (\ref{liouville}) by computing the divergence in the 
general case, and check that it is null if all diagonal games are zero-sum. 
This proves that in $N$-player polymatrix games with loops, as long as loops
are antisymmetric games, the \emph{quantity of information} is preserved in 
cumulative payoffs space. It also means that it will be hard to converge; 
for instance, no interior rest point cannot be locally attractive. Indeed, if that was 
the case, it would mean that locally, the volume would shrink around the rest point.

\subsection{Volume conservative games are zero-sum}

What is even more interesting is that the inverse statement is also 
true. The preservation of information in cumulative payoff space is 
specific to zero-sum diagonal games. We do not mean that a non-zero-sum game
cannot preserve volume at some points, but rather that preserving volume 
at many points implies the zero-sum property. The precise number of points
can be controlled by combinatorial Nullstellensatz arguments, and more
precisely, the relation between a multivariate polynomial and the geometry
of its vanishing set.
  
The argument is that the divergence of the vector field (in cumulative 
space) is a multivariate polynomial, and in particular, this polynomial 
vanishes exactly at points were the volume is preserved.

We give a proof for the evolutionary game theory settings, but it can be 
easily transported to $N$-player polymatrix games, notations would merely 
become heavier.
\begin{theorem}
  \label{vol-zero-sum}
  Consider a 1-player game $\Gamma$. 
  Let $\phi$ be the flow of replicator dynamics 
  and $\psi$ its diffeomorphic conjugate onto the cumulative payoffs space. 
  If there exists 
  an open set $U$ of $\R^{|\mathcal A|-1}$ such that $\psi$ is volume 
  conservative at any point of $U$, then $\Gamma$ is equivalent to a 
  1-player zero-sum game (summation of an antisymmetric matrix and a matrix of the form 
  $
  \begin{pmatrix} 1 & \cdots & 1 \end{pmatrix}^T\begin{pmatrix} c_{1} & \cdots & c_{n} \end{pmatrix}$).
\end{theorem}

\begin{proof}
  Let us rewrite $\mathcal A = \{1 \ldots n\}$, and
  $A = (A_{\alpha, \beta})$ the matrix corresponding to $\Gamma$.
  The existence of such an open set $U$ induces another \emph{interior} open set
  $V$ of $\mathcal X$ such that for any point $x$ of $V$,
  the divergence of $\frac{\partial \psi}{\partial t}$ at $y = \f(x)$ is null. 
  A general formula for this divergence is given in the appendix (see 
  Lemma~\ref{div-formula}). Using it, we get
  \begin{equation}
  \label{monomial}
  \divergence \;\frac{\partial \psi}{\partial t}(y) =
  \sum_{\alpha=1}^n x_\alpha A_{\alpha, \alpha} 
  - \sum_{\alpha=1}^n \sum_{\beta=1}^n x_\alpha x_\beta A_{\alpha, \beta}
  \end{equation}
  The above equality means that for any $x$ of $V$,
  $\sum_\alpha x_\alpha A_{\alpha, \alpha} = 
  \sum_{\alpha, \beta} x_\alpha x_\beta A_{\alpha, \beta}$. Take any 
  action $\gamma \in \{1 \ldots n\}$, without loss of generality action
  $\gamma = n$. We have $x_n = 1 - x_1 - \ldots - x_{n-1}$. Saying that the 
  divergence of $\frac{\partial \psi}{\partial t}$ is null on $V$
  means that the multivariate polynomial (\ref{monomial}) of 
  $[x_1 \ldots x_{n-1}]$ vanishes 
  on the open set $V$. Since $V$ is interior and open, the latest polynomial 
  have to be null \cite{alon1999combinatorial}. 
  Accordingly, all the coefficients of (\ref{monomial}) are
  zero's.

  In particular, the coefficient of $x_\alpha^2$ is null; but developing 
  (\ref{monomial}), this coefficient is precisely 
  $A_{\alpha, \alpha} + A_{n, n} - A_{\alpha, n} - A_{n, \alpha}$. As we took 
  $\gamma = n$ \emph{wlog}, we just proved that $A_{\alpha, \alpha} + 
  A_{\beta, \beta} - A_{\alpha, \beta} - A_{\beta, \alpha} = 0$, that is
  \begin{equation}
  \label{equiv-zs}
    A_{\alpha, \alpha} - A_{\alpha, \beta} 
    = - (A_{\beta, \beta} - A_{\beta, \alpha})
  \end{equation}
  for any actions $\alpha$ and $\beta$. This means that $A$ can be rewritten 
  as the sum of an antisymmetric matrix $B$ and a column-constant matrix. 
  Specifically,
  \begin{equation}
    A = B + \begin{pmatrix} 1 & \cdots & 1 \end{pmatrix}^T
    \begin{pmatrix} A_{1, 1} & \cdots & A_{n, n} \end{pmatrix}
  \end{equation}
  with $B$ antisymmetric. It is easy to see \cite{Hofbauer98} that
  the flow of replicator dynamics with matrix $B$ is the same as the one
  with matrix $A$, hence (from the perspective of replicator dynamics) $\Gamma$ is equivalent to a zero-sum game.
\qed
\end{proof}

One is easily convinced that this is generalizable to much more general games,
for e.g. polymatrix games,
by adapting the proof the following way: the multivariate polynomial's variables 
are strategies $x_{i, \alpha}$,
and since the divergence of the vector field is separable on each players,
we get the exact same condition for each diagonal game. Therefore, we have
the more stronger result.

\begin{theorem} 
\label{coro-vol}
A $N$-player (polymatrix)\footnote{The theorem straightforwardly extends to any game that can be rewritten as the sum of a $N$-player game in
normal form and self-edges games (even without the polymatrix condition).} game  
 is volume conservative
in cumulative payoffs space if, and only if its diagonal games are equivalent
to zero-sum games (summation of an antisymmetric matrix and a matrix of the form $
  \begin{pmatrix} 1 & \cdots & 1 \end{pmatrix}^T\begin{pmatrix} c_{1} & \cdots & c_{n} \end{pmatrix}$).
\end{theorem} 

This formaly shows that volume preservation strongly
correlates with zero-sum games. Furthermore, a polymatrix game that conserves 
volume on an open set \emph{has} to conserve volume everywhere. Observe that 
we could have been less restrictive on the assumption relating the geometry
of the vanishing set, so there is room to improve this result.
The \emph{take home idea} may be \emph{if diagonal games
are not zero-sum, the volume cannot be preserved at too many points}.

\section{Limit behavior: 
Poincar\'e recurrence, cycles and convergence to boundary}

\subsection{Zero-sum games with interior Nash are Poincar\'e recurrent}

We generalize previous result from \cite{piliouras2014optimization}. 
It is already known that zero-sum
polymatrix games with no loops are volume conservative, and that they exhibit
Poincar\'e recurrence behavior when there exists an interior Nash equilibrium.
In fact, this is also true for polymatrix games allowing self-loops. The 
proof is the same in its structure, but the existence of self-loops requires to 
use different arguments. The volume preservation is already given by 
Theorem~\ref{volcons-general} from previous section. The idea is to prove that,
under the assumption of the existence of an interior Nash, the Kullback-Leibler
divergence is a constant of motion, and that this implies that every orbit is 
bounded in cumulative payoffs space. Then, the Poincar\'e recurrence theorem applies.

\begin{theorem}
  \label{general-thm}
  Consider a $N$-player zero-sum polymatrix game with self-loops. 
  Assume there exists 
  an interior Nash equilibrium, then replicator dynamics is Poincar\'e recurrent.
\end{theorem}

\noindent
The proof relies on two key lemmas proven in appendix
(see \ref{lem-klinv}, \ref{lem-bounded-away}).

\begin{lemma}
  \label{klinv}
  Under the same assumptions and given $x^*$ an interior Nash equilibrium, 
  the sum of Kullback-Leibler divergences $\sum_{i = 1}^N \kl(x_i^*\| x_i)$
  is a constant of motion.
\end{lemma}

\begin{lemma}
  \label{bounded-away}
  Under the same assumption and given $x^*$ an interior Nash equilibrium, the
  sum, for any interior point $x \in \interior \; \mathcal X$, its orbit $\gamma 
  = \{\phi(x, t) \text{ : } t \ge 0\}$ is bounded away from the boundary.
\end{lemma}

\begin{proof}[Theorem~\ref{general-thm}]
Then the theorem follows directly from Poincar\'e recurrence theorem. The volume
is preserved in cumulative payoffs space, while every orbit stays bounded. 
Hence, the system is Poincar\'e recurrent in cumulative payoffs space; and
this property is transported to strategy space via the diffeomorphism
$\f^{-1}$.
\qed\end{proof}

\begin{remark}
To show Theorem~\ref{general-thm}, 
we used the fact that $\kl(x^* \| -)$ is a constant of 
motion. This property does not hold in general, but this is not the important
point; what is cirtical is that orbits \emph{remain bounded}.
The conservation of $\kl$ is no more than a tool to show this very property.
\end{remark}

\subsection{Poincar\'e recurrence and evolutionary game theory}

Given any polymatrix game, either there exists an interior Nash equilibrium,
or  no interior point is an equilibrium. The first case has been
dealt with. As far as the second case is concerned, 
previous work \cite{mertikopoulos2017cycles} 
have shown that in the 2-players case, the absence
of interior Nash equilibria enforces orbits to collapse to boundary. We show
that this is also true for 1-player zero-sum games (i.e., for evolutionary
game theory). Although not using the language of information theory, the results 
about the existence of strict Lyapunov functions and collapse to the boundary were first developed in 
\cite{Akin84}. Here we provide arguments to reduce this case to the more well studied two agent zero-sum games.
In combination with our Poincar\'e recurrence results, this will result in a complete picture of all possible limit behaviors of
 the system.

\begin{lemma}[\cite{Akin84}]
  \label{partial-support}
  Let $A$ be the matrix of a 1-player zero-sum game with no interior Nash
  equilibrium. Let $x^*$ be a Nash equilibrium of maximal support. 
  Then for any $x \in \mathcal X$ in the interior, 
  $\frac d{dt} \kl(x^* \| x) < 0$.
\end{lemma}

Our argument will use the 2-players result by giving a 2-players equivalent 
formulation of the 1-player game. Let $\Gamma^1$ be a 1-player zero-sum game
with $A$ the corresponding antisymmetric matrix. We claim
that $\Gamma^2$, the 2-player polymatrix game with matrices $A^{1, 2} = A$ and 
$A^{2, 1} = - A^T$ is equivalent to $\Gamma^1$ in the following way. There is
a canonic bijection between Nash equilibria of $\Gamma^1$ and \emph{symmetric}
Nash equilibria of $\Gamma^2$, that is, $x^*$ is an equilibrium of $\Gamma^1$
if, and only if $(x^*, x^*)$ is an equilibrium of $\Gamma^2$.

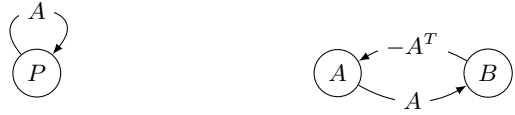
\begin{figure}[H]
  \centering
  \begin{tikzpicture}
    \node[draw, circle] (P) at (-3, 0) {$P$};
    \node[draw, circle] (A) at (1, 0) {$A$};
    \node[draw, circle] (B) at (3, 0) {$B$};
    
    \draw[->, >=latex] (A) to[bend right]
    node[pos=.5, fill=white] {\color{black} $A$} (B);
    \draw[->, >=latex] (B) to[bend right]
    node[midway, fill=white] {\color{black} $-A^T$} (A);
    \draw[->, >=latex] (P) to[min distance=10mm, in=90-40, out=90+40] 
    node[midway, fill=white] {\color{black} $A$} (P);
  \end{tikzpicture}
  \caption{The equivalent 2-players formulation}
\end{figure}

Moreover, it is easy to show that the diagonal 
$D = \{(x, x) \text{ : } x \in \Gamma^1(\mathcal X)\}$ is a stable space of 
$\Gamma^2$ under replicator dynamics, and that its canonic projection 
gives back exactly $\Gamma^1$. Now, if there is no interior equilibrium for
$\Gamma^1$, there cannot be interior symmetric equilibria for $\Gamma^2$. The
key point will be to show that there cannot be interior equilibria \emph{at all} 
for $\Gamma^2$.

\begin{lemma}
  Assume $\Gamma^1$ has no interior Nash. Then, $\Gamma^2$ has no interior Nash.
\end{lemma}

\begin{proof}
  We prove it by contradiction. Assume there exists an interior Nash equilibrium
  in $\Gamma^2$, say $(x^*, y^*)$. It is well known that for each agent the set of equilibrium strategies 
  coincides with their maxmin strategies which only depend on the agent's own payoff matrix. 
  Since both agents share the same payoff matrix, $x^*, y^*$ are maxmin strategies for both agents and 
  hence $(x^*, x^*)$  is also an (interior) Nash of $\Gamma^2$. But this symmetric state is immediately a
  Nash equilibrium for $\Gamma^1$ as well and we have reached a contradiction.
\qed\end{proof}

\begin{proof}[Proposition~\ref{partial-support}]
  If the 1-player game of antisymmetric $A$ has no interior Nash, then
  its 2-player equivalent game has no interior Nash either. To avoid ambiguities,
  write $\phi_1$ the flow of (\ref{rd}) of the 1-player game and 
  $\phi_2$ the flow of (\ref{rd}) of the 2-player one. Writing $x^*$ 
  a Nash equilibrium of maximal support of the 1-player version, for any 
  strategy $x$ of full-support, previous results \cite{mertikopoulos2017cycles} 
  guarantee
  \begin{equation}
     \frac d{dt} \bigg [\kl(x^* \| \phi_2((x, x), 0)_1) + \kl(x^* \| 
    \phi_2((x, x), 0)_2) \bigg] < 0
  \end{equation}
  under the replicator dynamics on the 2-player equivalent game. As on the
  diagonal, this dynamic is exactly the one the 1-player game, we conclude
  \begin{equation}
    \frac d{dt} \bigg [ \kl(x^* \| \phi_1(x, 0)) \bigg ] < 0
  \end{equation}
\qed\end{proof}

\begin{theorem}[\cite{Akin84}]
  \label{collapses}
  Let be a 1-player zero-sum game with matrix $A$ and with no interior Nash
  equilibrium, on which we write $\phi$ the flow of (\ref{rd}). 
  Let $x^*$ be a Nash equilibrium. 
  Then for any interior point $x \in \mathcal X$, the orbit $\gamma
  = \{\phi(x,t) \text{ : } t \ge 0\}$ collapses to boundary. More precisely,
  for all $y \in \omega(x)$, $\supp \; y \subseteq \supp \; x^*$. 
\end{theorem}

The proof follows from standard Lyapunov arguments. For completeness, we provide the proof in Appendix \ref{thm-collapses}.
This theorem shows that in the absence of any interior equilibrium, every interior
orbits collapses to the face spanned by $\supp(x^*)$ with a $x^*$ of maximal 
support. It tells nothing about the behavior of orbits when coming close to 
this face. Do we have convergence, or do we get (Poincar\'e) recurrence/cycles on the 
boundary? In general, both are possible, depending on the initial condition.

\begin{figure}[H]
  \centering
  \begin{subfigure}{.5\linewidth}
    \centering
    \includegraphics[width=\linewidth]{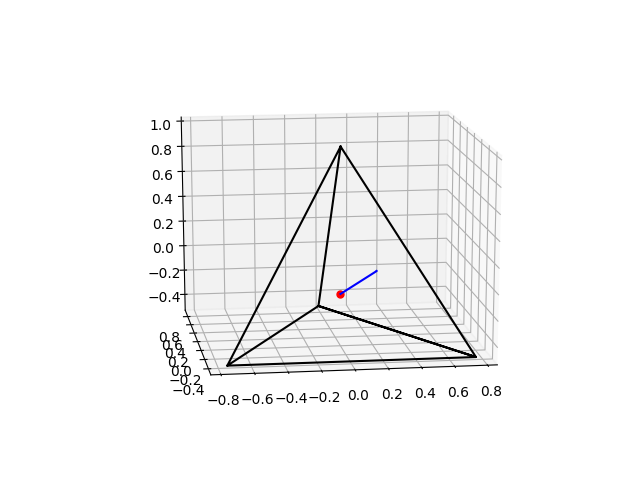}
  \end{subfigure}%
  \begin{subfigure}{.5\linewidth}
    \centering
    \includegraphics[width=\linewidth]{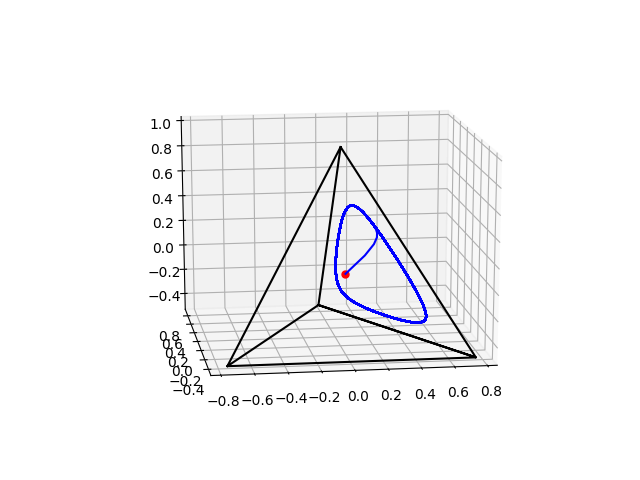}
  \end{subfigure}
  \caption{Converging and non-converging orbits in the same game.}
\end{figure}

Consider Rock-Paper-Scissor to which we add a dummy action, say Fork, which
scores -10 against any other action (excepted Fork itself). That is, consider
the 1-player zero-sum game with matrix
\begin{equation*}
  A = \begin{pmatrix} 0 & -1 & 1 & 10 \\
  1 & 0 & -1 & 10 \\
  -1 & 1 & 0 & 10 \\
  -10 & -10 & -10 & 0 \\ \end{pmatrix}
\end{equation*}
The Nash equilibrium is unique and $(\frac 13, \frac 13, \frac 13, 0)$. If one
starts at $(\frac 14, \frac 14, \frac 14, \frac 14)$, one converges to it. If
one starts at $(\frac 3{16}, \frac 5{16}, \frac 14, \frac 14)$, one collapses
to a periodic orbit on the boundary.

Combining the results we have so far, we can prove a fairly complete theorem
relating volume conservation, Poincar\'e recurrence and evolutionary game 
theory.

\begin{theorem}
  Let be a 1-player matrix game $A$ under the flow of \emph{replicator 
  dynamics}. The volume is preserved in cumulative payoffs space if, 
  and only if the game is equivalent to a zero-sum game; more precisely,
  if, and only if $A$ can be written as
  $A = B + (1 \; \cdots \; 1)^T (A_{1,1} \; \cdots \; A_{n, n})$ with
  $B$ an antisymmetric matrix.

  If that is the case, interior orbits exhibits Poincar\'e recurrent behavior
  if, and only if there exists an interior Nash equilibrium. If there is no
  interior equilibrium, every interior orbit collapses to the face
  spanned by the support of a Nash equilibrium of maximal support.
\end{theorem}

\begin{proof}
  Let $A$  be a single-agent matrix game under replicator dynamics. Assume
  the volume is conserved in cumulative payoff space. Then, by Theorem 
  \ref{vol-zero-sum}, $A$ is equivalent to a zero-sum game ($A = B + (1 \; \cdots \; 1)^T (A_{1,1} \; \cdots \; A_{n, n})$ where
  $B$ is an antisymmetric matrix.). Conversely,
  if $A$ is equivalent in the above sense to a zero-sum game, one can assume without loss of
  generality that $A$ is antisymmetric. Then, by Theorem~\ref{volcons-general},
  the volume is preserved at any point. This proves the first part of the theorem.

  Now, assume $A$ is antisymmetric. If there exists an interior Nash, by
  Theorem~\ref{general-thm}, the system is Poincar\'e recurrent. Conversely,
  if the system is Poincar\'e recurrent, there has to exist an interior Nash.
  Assume on the contrary that there is no such equilibrium. Let $x^*$ be a
  Nash equilibrium. Consider the open ball $U = B(\frac 1n (1 \ldots 1), \epsilon)$
  with $\epsilon > 0$ small. We know that there exists an orbit $\gamma$
  intersecting $U$ infinitely often. If $\epsilon$ is small enough, 
  by taking $x$ any point of $\gamma$, that means that 
  \begin{equation}
  \label{limsup-pos}
  \limsup [\dist(\phi(x, t), \textnormal{bd}(\mathcal X))]> 0
  \end{equation}
  But by Theorem
  \ref{collapses}, $\gamma$ should collapses to the boundary. This contradicts
  (\ref{limsup-pos}).
\qed\end{proof}


\section{Cycles in dimension 3}

In this section, we give a proof that the flow $\phi$ of replicator dynamics
is periodic for every interior initial condition of 1-player zero-sum games of dimension 3
with interior Nash equilibrium. The proof uses the 
 Poincar\'e-Bendixson
Theorem, that we recall here.

\vspace{1em}
\begin{theorem}[Poincar\'e-Bendixson]
  A limit set $\omega(x)$ of a $\mathcal C^1$ dynamical system over the plane,
  if non-empty and compact, that does not contain a rest point is a periodic
  orbit.
\end{theorem}

In the following, we make the assumption that the game is a 1-player
zero-sum game of dimension 3 that has an interior Nash equilibrium $x^*$.

\begin{lemma}
  \label{klconst}
  Let be a 1-player zero-sum game with matrix $A$ with an interior Nash 
  equilibrium. Then, for any interior initial condition $x$, the Kullback-Leibler divergence 
  to any interior Nash equilibrium $x^*$
  is constant over the limit set $\omega(x)$. More precisely, for any 
  $y \in \omega(x)$, we have $\kl(x^* \| y) = \kl(x^*\|x)$.
\end{lemma}

\begin{proof}
  $\kl$ is continuous defined on the compact set $\mathcal X
  = \Delta \R^n$, so is uniformly continuous. What is more, by compactness of 
  $\mathcal X$, $\lim_{t \rightarrow \infty}\dist(x(t), \omega(x)) = 0$.
  Therefore, since $\kl(x^* \| -)$ is a constant of motion by Lemma~\ref{klinv},
  we prove that for any $y \in \omega(x)$, $\kl(x^* \| y) = \lim 
  \kl(x^* \| x(t_n)) = \kl(x^* \| x)$.
\qed\end{proof}

\begin{lemma}
  \label{omega-cycle}
  Let be a 1-player zero-sum game of dimension 3 with matrix $A$.
  Assume there exists an interior Nash. Then, for any interior point $x$, 
  $\omega(x)$ is a periodic orbit.
\end{lemma}

\begin{proof}
 This system has clearly two degrees of freedom and is hence planar.
  Let $x^\star$
  be any interior Nash equilibrium. By Theorem~\ref{general-thm}, 
  $\kl(x^\star\| -)$ is a constant of motion. If the initial point is not an 
  equilibrium, the orbit is non-trivial and $\kl(x^\star \| \phi(x, -))$
  is constant and positive. Therefore, $\phi(x, -)$ is bounded away from any 
  rest point of the flow. It follows that $\omega(x)$ does not contain any rest 
  point and the statement follows by applying the Poincar\'e-Bendixson's theorem. 
\qed\end{proof}

\begin{figure}[H]
  \centering
  \includegraphics[width=\linewidth]{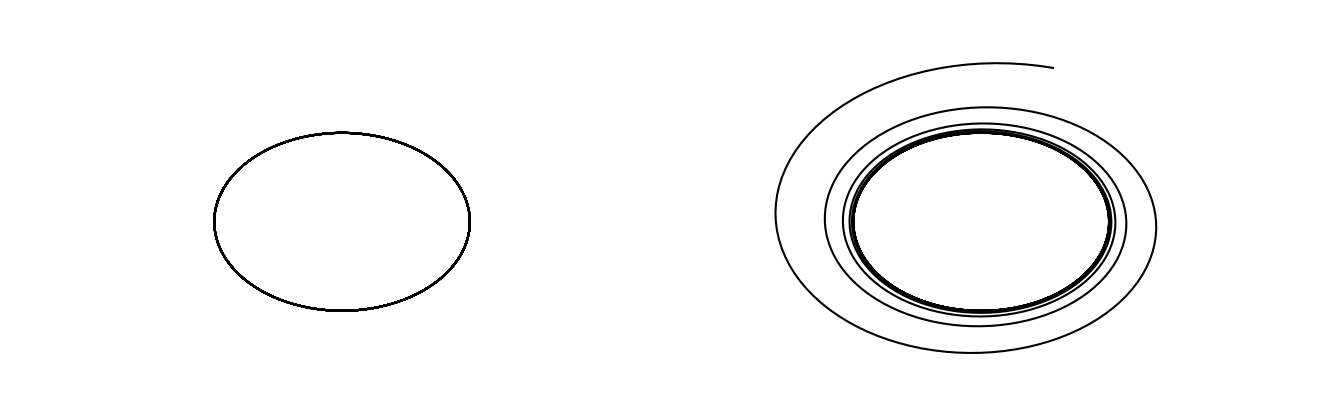}
  \caption{Up to now, we know that interior orbits should act like periodic 
  orbits in the limit. In general, it means that they show a spiral-like 
  behavior, converging to a Jordan's curve. The point is to show these spirals
  are precisely periodic orbits, i.e., that all these spirals are 
  ellipsoid-like.}
\end{figure}

\vspace{1em}
\begin{theorem}
  \label{periodic}
  Let be a 1-player zero-sum game of dimension 3 with matrix $A$.
  Assume there exists an interior Nash. Then,
  any interior point $x$ belongs to a periodic orbit.
\end{theorem}

\begin{proof}
  The result is obvious for interior equilibria. Assume $x$ is not an
  equilibrium. 
  By Lemma~\ref{omega-cycle}, $\omega(x)$ is a periodic orbit. It means that geometrically, $\omega(x)$ is a Jordan's curve
  of the plane $\Delta : x_1 + x_2 + x_3 = 1$, so $\Delta$ is separated 
  into two connected components; an interior $A$ and an exterior $B$. 
    
  By assumption, we know that there is a rest point in the interior
  of $\mathcal X$. What is 
  more, we also know that the time average of the strategy over $\omega(x)$
  is an equilibrium\footnote{This follows immediately from the no-regret property of replicator \cite{mertikopoulos2017cycles}.}, that will lie in the convex hull, hence interior to 
  $\mathcal X$. Let us denote it 
  $x^\star$. We claim that this $x^\star$ has to be a point of $A$. Assume,
  on the contrary, that $x^\star \notin A$. Because $\omega(x)$ is non-trivial,
  $A$ is non-empty. Let $a \in A$. Draw a semi-infinite ray starting from
  $x^\star \in B$ in the direction of $a\in A$. Because $A$ is bounded, this ray
  will transit from $B$ to $A$ then $A$ to $B$ at least once. Hence, it 
  cross $\omega(x)$ at least two times, say first $\omega_1$ then $\omega_2$.
  But, $\kl(x^\star \| \cdot)$ is a strict convex function, globally minimal at 
  $x^\star$, so will stricly increase as one advance on the ray. Therefore,
  $\kl(x^\star \| \omega_1) < \kl(x^\star \| \omega_2)$, which contradicts
  Lemma~\ref{klconst}. Accordingly, $x^\star \in A$.

  \begin{figure}[H]
    \centering
    \begin{tikzpicture}
      \draw[line width=.75, fill=black!10, black!10] (0, 0) ellipse(2 and 1.5);
      \begin{scope}
      \clip (-2.1, 1.5) -- (0.5, 1.6) -- (-2.1, -1.5) -- cycle;
      \draw[line width=.75, fill=black!20] (0, 0) ellipse(2 and 1.5);
      \end{scope}
      \draw[fill=black!10, black!10] 
      (-1.6, -.9125) to[in=-160, out=90] (.4, 1.475) -- (-1.6, -.9125);
      \draw[line width=.75] (-1.6, -.9125) to[in=-160, out=90] (.4, 1.475);

      \draw (-.4, .3) -- (-2.6, 1.4);
      \draw[->] (-.4, .3) -- (-2.0, 1.1);
      
      \draw[fill=black] (-.4, .3) circle(0.025cm);
      \draw[fill=black] (-1.01, .605) circle(0.025cm);
      \draw[fill=black] (-1.6, .9) circle(0.025cm);
      \node (a) at (-.4+.25, .3-.25) {$x^\star$};
      \node (a) at (-1.01+.3, .605) {$\omega_1$};
      \node (a) at (-1.6-.1, .9+.25) {$\omega_2$};
      \node (a) at (-1.75, 0) {$A$};
      \node (a) at (-2.5, 0.25) {$B$};
      \node (a) at (-1, -.25) {$B$};

      \node (x_star) at (-0.33+6, 0.1) {$x^\star$};
      \draw[fill=black] (0+6, 0) circle(0.025cm);
      \draw[line width=.75] (0+6, 0) ellipse(2 and 1.5);
      \draw (0+6,0) to (3+6, 1.5);
      \draw[->] (0+6,0) to (2.5+6, 1.25);
      \node (l) at (2.4+6, 1) {$L$};

      \draw[fill=black] (2*0.83548784338+6, 1.5*0.549508929) circle(0.025cm);
      \node (lol) at (2*0.83548784338+6+0.075, 1.5*0.549508929+0.3) {$x_\omega$};
    \end{tikzpicture}
    \caption{Geometric visualisation of the proof}
  \end{figure}
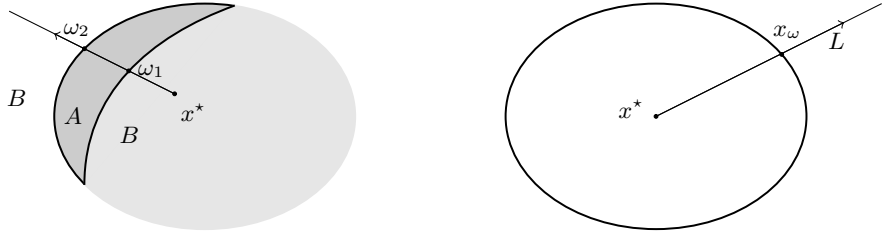

  To summarize, $\omega(x)$ is a periodic orbit, and its time average,
  an interior Nash equilibrium $x^\star$, lies in the interior of $\omega(x)$,
  $A$. We want to prove that the orbit starting from $x$ is a periodic orbit.
  To prove that, we show that $x \in \omega(x)$.
  From $x^\star$, 
  fire a semi-infinite line $L$ from $x^\star$ in the direction of $x$. 
  Because $A$ is bounded, this semi-infinite line have to cross $\omega(x)$
  in at least a point. Choose one of those and call it $x_\omega$. We claim that
  $\kl(x^\star \| - )$ is equal to $\kl(x^\star \| x)$ on $L$ only
  at $x_\omega$. Indeed, $\kl(x^\star \| -)$ is a strict convex
  function with minimum at $x^\star$, so it is stricly growing as one
  advances along the straight line $L$. As a consequence, $x_\omega$ is the 
  unique intersection point between $L$ and $\kl^{-1}(x^\star \| - )
  [\kl(x^\star \| x)]$. 
  Accordingly, $x = x_\omega \in \omega(x)$.
\qed\end{proof}

\begin{proposition}
  Let be a one-player zero-sum game with $n = 3$ actions. 
  The following statements are equivalent:
  \begin{itemize}
    \setlength\itemsep{0pt}
    \item[(i)] there exists an interior Nash equilibrium
    \item[(ii)] there exists an interior cycle orbit 
    \item[(iii)] any orbit containing an interior point is an interior cycle.
  \end{itemize}
\end{proposition}
 

\begin{remark}
  This proof relies on the Kullback-Leibler divergence. That is, in a Poincar\'e
  recurrent system, we used an argument specific to game theory to show 
  that all interior orbits are periodic. Thinking of what Poincar\'e recurrent
  means, one may hope to get rid of the game theoretic proof and give a topological
  proof. The motivation is clear; for any open set, almost every orbit goes back
  arbitrarily close to its initial condition, and in addition, infinitely often.
  Therefore, we get what looks like a dense set of periodic orbits.

  That is, if a point is not a rest point, because we are in dimension two, 
  its orbit is infinitely-closely trapped between periodic orbits. There, we 
  claim that there is no hope to conclude that this orbit must be periodic with
  topological arguments only. Look at the counter-example below.

  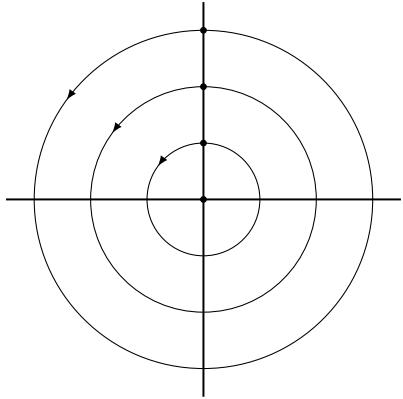
\begin{figure}[H]
    \centering
    \begin{tikzpicture}[scale=.75]
      \draw[line width = .75] (-3.5, 0) -- (3.5, 0);
      \draw[line width = .75] (0, -3.5) -- (0, 3.5);
      \draw[fill=black] (0, 0) circle(0.05cm);
      \draw[fill=black] (0, 1) circle(0.05cm);
      \draw[fill=black] (0, 2) circle(0.05cm);
      \draw[fill=black] (0, 3) circle(0.05cm);
      \draw[>=latex, 
            decoration={markings, mark=at position .4 with {\arrow{>}}},
            postaction={decorate}]
      (0, 0) circle(1);
      \draw[>=latex, 
            decoration={markings, mark=at position .4 with {\arrow{>}}},
            postaction={decorate}]
      (0, 0) circle(2);
      \draw[>=latex, 
            decoration={markings, mark=at position .4 with {\arrow{>}}},
            postaction={decorate}]
      (0, 0) circle(3);
    \end{tikzpicture}
    \caption{On the complex plane, consider the ODE $\dot z = \textbf i \; z$.
    The corresponding flow is $\phi(z, t) = z\cdot e^{\textbf it}$. 
    Hence, every orbits
    are circles, excepted the single rest point at the origin. Add the 
    \emph{velocity regularizer} 
    $\delta : z \mapsto \min \{1, \dist(z, \textbf i \N)\}$. The ODE becomes
    $\dot z = \textbf i \; \delta(z) z$. Then, almost all
    orbits are still circles, so there is a dense set of periodic orbits and
    the system is Poincar\'e recurrent. 
    Yet, if a point $z$ has integer module, it is arbitrarly close 
    to a periodic orbit, and its limit set is the rest point 
    $\textbf i |z|$.}
  \end{figure}
\end{remark}

\section*{Acknowledgments}

Georgios Piliouras acknowledges MOE AcRF Tier 2 Grant 2016-T2-1-170, grant PIE-SGP-AI-2018-01 and NRF 2018 Fellowship NRF-NRFF2018-07. This work was partially done while Victor Boone was a visitor at SUTD under the supervision of Georgios Piliouras. Victor Boone thanks Bruno Gaujal and Panayotis Mertikopoulos for helping to arrange the visit and for their overall guidance and mentorship.

\bibliography{ms}
\bibliographystyle{wine-class/splncs04}

\newpage

\appendix

\section{Proof of Theorem~\ref{volcons-general}}

The proof makes  use of the explicit formula of the 
diffeomorphism $\f^{-1}$ from cumulative payoffs space to strategy space. 
For each player $i$ and vector $y_i = (y_{i,1} \ldots y_{i,n_i-1})$ of 
cumulative payoffs, we recall from 
\cite{piliouras2014optimization,Hofbauer98} 
that the corresponding strategy 
$x_i = \f^{-1}(y_i)$ is explicitly
\begin{equation}
   (x_{i,1} \ldots x_{i, n_i}) = (\f^{-1}_1(y_i) \ldots \f^{-1}_{n_i}(y_i))
   = \left(\frac 1{S(y_i)}, 
   \frac {e^{y_{i, 1}}}{S(y_i)},
   \ldots,
   \frac {e^{y_{i, n_i-1}}}{S(y_i)}\right)
\end{equation}
where $S(y_i) = 1 + \sum_{\alpha=1}^{n_i-1} e^{y_{i, \alpha}}$.

\begin{lemma}
  \label{div-formula}
  Let $\phi$ be the flow of replicator dynamics (\ref{rd}) with $N$ agents.
  Let $\psi(y, -) = \f(\phi(\f^{-1}(y), -))$ be the diffeomorphic flow onto 
  cumulative payoffs space. Then, its divergence is 
  \begin{equation}
    \divergence \; \frac{\partial \psi}{\partial t}(y) = \sum_{i=1}^N
    \sum_{\alpha=1}^{n_i}
    \sum_{\beta = 1}^{n_i} x_{i,\alpha} x_{i, \beta} \left(
    \frac{\partial u_{i,\alpha}}{\partial x_{i,\alpha}} - 
    \frac{\partial u_{i,\alpha}}{\partial x_{i, \beta}}
    \right) 
  \end{equation}
\end{lemma}

\begin{proof} We want to compute the divergence of the vector field in
cumulative costs space, that is, 
$\sum_{i=1}^N \sum_{\alpha=1}^{n_i-1} \frac{\partial}{\partial y_{i,\alpha}}
\left[u_{i,\alpha+1}(\f^{-1}(y)) - u_{i, 1}(\f^{-1}(y))\right]$. Using the 
chain rule, this precisely is
\begin{equation}
  \label{div-term-scalar}
  \frac{\partial}{\partial y_{i,\alpha}} \left[u_{i,\beta}
  (\underbrace{\f^{-1}(y)}_{x(t)})
  \right] = \left \langle \nabla_{x(t)} u_{i, \beta} \; ,\;
  \frac \partial{\partial y_{i,\alpha}} \f^{-1}(y) \right \rangle
\end{equation}
Recall that the divergence is given by the sum over $i$ and $\alpha$ of terms
expressed in equation (\ref{div-term-scalar}). If there is no loop, we do not 
need
to compute this scalar product since it is over vectors of disjoint supports; 
but this is not true in general. So, we start by computing the right term of 
(\ref{div-term-scalar}). 
In the three following equations,
we write $S$ for $1 + \sum_{\gamma=1}^{n_i-1} e^{y_{i,\gamma}}$. We check that
\begin{equation}
  \label{f1}
  \frac \partial {\partial y_{i,\alpha}} \bigg[\f^{-1}_1(y)\bigg]
  = \frac \partial {\partial y_{i,\alpha}} \left ( \frac 1{1 + \sum_\gamma
  e^{y_{i,\gamma}}} \right) = - \frac{e^{y_{i,\alpha}}}{S^2} = 
  - x_{i,\alpha+1} x_{i, 1}
\end{equation}
\begin{equation}
  \label{fl}
\text{for } \gamma \ne \alpha,\quad \frac \partial {\partial y_{i,\alpha}} 
\bigg [\f^{-1}_{\gamma+1}(y)\bigg]
= \frac \partial {\partial y_{i,\alpha}} \left(\frac{e^{y_{i,\gamma}}}
{1 + \sum_\gamma e^{y_{i,\gamma}}}
\right) =
- \frac{e^{y_{i,\alpha}} e^{y_{i,\gamma}}}{S^2} = -x_{i, \alpha+1} 
x_{i, \gamma+1}
\end{equation}
\begin{equation}
  \label{fj}
\text{for } \gamma = \alpha,\quad \frac \partial {\partial y_{i,\alpha}} 
\bigg[\f^{-1}_{\gamma+1}(y)\bigg]
= \frac{e^{y_{i,\alpha}}} S - \left(\frac{e^{y_{i,\alpha}}} S\right)^2
= x_{i, \alpha+1} (1 - x_{i, \alpha+1})
\end{equation}
Injecting (\ref{f1}) (\ref{fl}) and (\ref{fj}) into (\ref{div-term-scalar}), 
we get to
\begin{align*}
  \frac{\partial}{\partial y_{i,\alpha}} \bigg[u_{i,\beta}\left(\f^{-1}(y)\right)
  \bigg] & = \sum_{\gamma=1}^{n_i} \frac{\partial u_{i,\beta}}{\partial 
  x_{i,\gamma}}
  \frac {\partial \f^{-1}_\gamma} {\partial y_{i,\alpha}} (y) \\
  &= \frac{\partial u_{i,\beta}}{\partial x_{i,\alpha+1}} 
  \bigg [ x_{i, \alpha+1} (1 - x_{i, \alpha+1}) \bigg ]
  - \sum_{\gamma \ne \alpha+1} \left[x_{i, \alpha+1} x_{i, \gamma}\frac{\partial 
  u_{i,\beta}}{\partial x_{i,\gamma}}\right]
\end{align*}
Recall that $x_i \in \mathcal X_i$, thus $1 - x_{i, \alpha+1} = 
\sum_{\gamma \ne \alpha+1} x_{i,\alpha}$. Using it in the above equality, we 
obtain
\begin{equation}
  \label{formula-k}
  \frac{\partial}{\partial y_{i,\alpha}} \bigg[u_{i,\beta}(\f^{-1}(y)) \bigg]
  = \sum_{\gamma \ne \alpha+1} x_{i, \alpha+1} x_{i,\gamma} \left(
  \frac{\partial u_{i,\beta}}{\partial x_{i,\alpha+1}} - 
  \frac{\partial u_{i,\beta}}{\partial x_{i,\gamma}}\right)
\end{equation}
Observe that the term for $\gamma = \alpha+1$ would be null. Thus, we can add it.
Then, doing (\ref{formula-k})$[\beta:=\alpha+1] - (\ref{formula-k})[\beta:=1]$, 
we get
\begin{equation}
  \frac{\partial \dot y_{i,\alpha}}{\partial y_{i,\alpha}}
  = \sum_{\beta=1}^{n_i} x_{i, \alpha+1} x_{i, \beta}
  \left[\frac{\partial u_{i,\alpha+1}}{\partial x_{i,\alpha+1}} - 
  \frac{\partial u_{i,1}}{\partial x_{i,\alpha+1}} +
  \frac{\partial u_{i,1}}{\partial x_{i,\beta}} -
  \frac{\partial u_{i,\alpha+1}}{\partial x_{i,\beta}}
   \right]
\end{equation}
Sum over $i \in \{1 \ldots N\}$ and $\alpha \in \{1 \ldots n_i - 1\}$. As 
the term for $\alpha = 0$ is null, we can do the change of variable 
$\alpha' := \alpha + 1$ and sum it from 1 to $n_i$. Therefore, the divergence is
\begin{equation}
\divergence \; \frac{\partial \psi}{\partial t}(y) =
\sum_{i=1}^N\sum_{\alpha=1}^{n_i}\sum_{\beta=1}^{n_i} 
  x_{i, \alpha} x_{i, \beta}
  \left[\frac{\partial u_{i,\alpha}}{\partial x_{i,\alpha}} - 
  \frac{\partial u_{i,\alpha}}{\partial x_{i,\beta}} +
  \frac{\partial u_{i,1}}{\partial x_{i,\beta}} -
  \frac{\partial u_{i,1}}{\partial x_{i,\alpha}} 
   \right]
\end{equation}
But $\sum_{\alpha=1}^{n_i}\sum_{\beta=1}^{n_i} x_{i, \alpha} x_{i, k\beta} 
\left[\frac{\partial u_{i,1}}{\partial x_{i,\beta}} - \frac{\partial u_{i,1}}
{\partial x_{i,\alpha}} \right]$ is antisymmetric, so cancels out. We are left 
with
\begin{equation}
  \label{div}
  \divergence \; \frac{\partial \psi}{\partial t}(y) =
  \sum_{i=1}^N\sum_{\alpha=1}^{n_i}\sum_{\beta=1}^{n_i} 
  x_{i, \alpha} x_{i, \beta}
  \left[\frac{\partial u_{i,\alpha}}{\partial x_{i,\alpha}} - 
  \frac{\partial u_{i,\alpha}}{\partial x_{i,\beta}}
  \right]
\end{equation}
\qed\end{proof}

Now, the proof of Theorem~\ref{volcons-general} is straightforward.

\begin{proof}[Theorem~\ref{volcons-general}]
  We assume that all $A^{i,i}$ are antisymmetric matrices.
  The partial derivative is 
  $\frac{\partial u_{i, \alpha}}{\partial x_{i, \gamma}} =
  (\frac{\partial A^{i, i} x_i}{\partial x_{i, \gamma}} )_\alpha
  = \frac\partial{\partial x_{i, \gamma}}\big[\sum_\beta A^{i,i}_{\alpha, 
  \beta} x_\beta\big]= (A^{i,i}_{\alpha, \gamma})_\alpha$. 
  Rewrite (\ref{div}).

\begin{equation}
  \label{eq-volcons-general}
  \divergence \; \frac{\partial \psi}{\partial t}(y) =
  \sum_{i=1}^N\sum_{\alpha=1}^{n_i}\sum_{\beta=1}^{n_i} \bigg[
  x_{i, \alpha} x_{i, \beta} (A^{i,i}_{\alpha,\alpha} - A^{i,i}_{\alpha,\beta})
  \bigg]
\end{equation}

But $A^{i,i}_{\alpha, \alpha} = 0$ for all $\alpha$, and 
for all $i$, $(x_{i, \alpha} x_{i, \beta} A^{i,i}_{\alpha, \beta})_{\alpha, 
\beta}$ is an antisymmetric 
matrix, and hence the sum of all its coefficient is null. Accordingly,
the right term of (\ref{eq-volcons-general}) is zero, and applying Liouville's
formula (\ref{liouville}), $\psi$ conserves volume.
\qed\end{proof}

\begin{remark}
  Note that nowhere in the proof did we actually use the assumption that 
  the $N$-player game without the self-edges has to be a graphical polymatrix
  game. The exact same proof would hold for any $N$-player normal form game
  with antisymmetric self-edges.
\end{remark}

\section{Proof of Lemma~\ref{klinv}}
\label{lem-klinv}

Computing the time derivative, we get that
  \begin{equation}
    \label{dt-kull}
    \frac d{dt} \left[\sum_{i = 1}^N \kl(x_i^*\| x_i)\right]
    = \sum_{i=1}^N u_i(x) - 
      \sum_{i=1}^N \sum_{\alpha_i = 1}^{n_i} x_{i,\alpha_i}^* u_{i,\alpha_i}(x)
  \end{equation}
  Now, recall that for zero-sum games, $\sum u_i(-) = 0$, since for any $x$,
  \begin{equation}
  \sum_{i=1}^N u_i(x) = \sum_{i=1}^N \sum_{j=1}^N \sum_{\alpha_i=1}^{n_i}
  \sum_{\alpha_j=1}^{n_j} x_{i, \alpha_i} x_{j, \alpha_j} A^{i,j}_{\alpha_i, 
  \alpha_j} 
  \end{equation}
  that is a sum of all the coefficients of the antisymmetric matrix
  $(A^{i,j}_{\alpha_i,\alpha_j} x_{i, \alpha_i} x_{j, \alpha_j})$
  \footnote{$A^{i, j}_{\alpha_i, \alpha_j} = - A^{j, i}_{\alpha_j, \alpha_i}$}
  hence is zero. Therefore, (\ref{dt-kull}) simplifies into
  \begin{align}
    \frac d{dt} \left[\sum_{i = 1}^N \kl(x_i^*\| x_i)\right]
    & = - \sum_{i=1}^N \sum_{\alpha_i = 1}^{n_i} x_{i,\alpha_i}^* 
    u_{i,\alpha_i}(x) \\
    & = - \sum_{i=1}^N \sum_{j=1}^N 
          \sum_{\alpha_i=1}^{n_i} \sum_{\alpha_j=1}^{n_j} 
          x^*_{i, \alpha_i} A^{i,j}_{\alpha_i, \alpha_j} x_{j, \alpha_j}  \\
    & = \sum_{j=1}^N \sum_{\alpha_j=1}^{n_j} x_{j, \alpha_j} \sum_{i=1}^N 
    \sum_{\alpha_i=1}^{n_i} A^{j ,i}_{\alpha_j, \alpha_i} x^*_{i, \alpha_i}
  \end{align}
  Then,
  \begin{align}
    \frac d{dt} \left[\sum_{i = 1}^N \kl(x_i^*\| x_i)\right]
    & = \label{jhzo} \sum_{i=1}^N \sum_{\alpha_i = 1}^{n_i} 
    x_{i,\alpha_i} u_{i}(x^*) \\
    & = \label{jhzo2}\sum_{i=1}^N u_{i, \alpha _i}(x^*) = 0
  \end{align}
  where we used one fundamental property of full-support Nash equilibria, that is 
  here $u_{i,\alpha_i}(x^*) = u_i(x^*)$, to go from (\ref{jhzo}) to 
  (\ref{jhzo2}).

\section{Proof of Lemma~\ref{bounded-away}}
\label{lem-bounded-away}
Let $x$ be an interior point that is not an equilibrium. Then, 
according to the previous lemma, for any $y \in \gamma$, 
$C = \sum_i \kl(x^*_i \| x_i) = \sum_i(x^*_i \| y_i) > 0$. But 
$\kl(x^*_i \| -) \ge 0$ for any player $i$, so $\kl(x^*_i \| y)$ is 
a point of the segment $[0, C]$ independently of the player $i$ and the 
point $y \in \gamma$. Expanding $\kl(x^*_i \| y)$, this means 
\begin{equation}
  x^*_{i, \alpha} \log y_{i, \alpha} \ge
  - C - h_2(x^*_i) - \sum_{\beta \ne \alpha} x^*_{i, \beta} \log y_{i, \beta}
\end{equation}
for any player $i$ and action $\alpha \in \mathcal A_i$, and where $h_2(x^*_i)$
is the Shannon entropy of $x^*_i$. Writing $C'_{i, \alpha} =
\sum_{\beta \ne \alpha} x^*_{i, \beta} \log y_{i, \beta}$, we check that 
$C'_{i, \alpha}$ is non-positive. Therefore
$x^*_{i, \alpha} \log y_{i, \alpha} \ge - C - h_2(x_i^*)$, that is
\begin{equation}
  y_{i, \alpha} \ge \exp \left(-\frac{C + h_2(x^*_i)}{x_{i, \alpha}^*}\right)
\end{equation}
Define $\delta = \min_{i \in \mathcal N, \alpha \in \mathcal A_i}
\exp \left(-\frac{C + h_2(x^*_i)}{x_{i, \alpha}^*}\right)$. Then, at any point 
$y \in \gamma$, for any player $i$ and action $\alpha \in \mathcal A_i$,
$y_{i, \alpha} \ge \delta > 0$. Accordingly, $\gamma$ is bounded away from the
boundary.

\section{Proof of Theorem~\ref{collapses}}
\label{thm-collapses}

  This is a typical Lyapunov function argument. Let $\mathcal X'$ denotes 
  the set of strategies $y$ of support strictly containing $\supp \; x^*$.
  In particular, this set contains the interior of $\mathcal X$. 
  Fix $x \in \mathcal X'$. 
  By Proposition~\ref{partial-support}, $\frac d{dt} \kl(x^* \| \phi(x, 0)) < 0$. 
  We claim that $\omega(x) \cap \mathcal X' = \varnothing$.

  We prove it by contradiction. Assume that $\omega(x) \cap \mathcal X'
  \ne \varnothing$. Accordingly, there exists $y \in \omega(x) \cap \mathcal X'$,
  and more precisely, there exists $t_n \uparrow + \infty$ such that 
  $\lim \phi(x, t_n) = y$.
  It is clear that the orbit $\gamma$ of $x$ is a subset of $\mathcal X'$.
  Therefore, $t \mapsto \kl(x^* \| \phi(x, t))$ is a decreasing function.
  It is thus immediate that for all $n$, 
  $\kl(x^* \| \phi(x, t_n)) \ge \kl(x^* \| y)$, and because $(t_n)$ is 
  a growing sequence diverging to infinity,
  \begin{equation}
    \label{jdklqs}
    \forall t \ge 0 \quad\quad \kl(x^* \| \phi(x, t)) \ge \kl(x^* \| y)
  \end{equation}
  Now, recall that $y \in \mathcal X'$, so $\frac d{dt} \kl(x^* \|
  \phi(y, 0)) < 0$. So, there exists $t' > 0$ such that 
  $\kl(x^* \| y) > \kl(x^* \| \phi(y, t'))$. In addition, using the continuity
  of $\phi$, $\lim \phi(x, t' + t_n) = \lim \phi(\phi(x, t_n), t')
  = \phi(\lim \phi (x, t_n), t') = \phi(y, t')$. But $\kl$ is continuous as
  well, so
  \begin{equation}
    \lim \kl(x^* \| \phi(x, t'+t_n)) = \kl(x^* \| \phi(y, t'))
  \end{equation}
  Yet, $\kl(x^* \| y) > \kl(x^* \| \phi(y, t'))$ so for $n$ large enough,
  \begin{equation}
    \kl(x^* \| \phi(y, t')) \le \kl(x^* \| \phi(x, t'+t_n)) < \kl(x^* \| y)
  \end{equation}
  This contradicts (\ref{jdklqs}).

\end{document}